\pdfoutput=1
\documentclass[letterpaper,10pt,conference]{ieeeconf}

\usepackage[T1]{fontenc}
\usepackage{amsmath,amssymb,amsfonts,mathtools}
\usepackage{bm}
\usepackage{xcolor}
\usepackage{graphicx}

\IEEEoverridecommandlockouts

\usepackage{amsthm}
\let\labelindent\relax
\usepackage{enumitem}

\DeclareMathOperator*{\argmin}{arg\,min}
\DeclareMathOperator{\col}{col}
\newcommand{\eqdef}{\mathrel{\mathop:}=}
\allowdisplaybreaks

\newtheorem{theorem}{Theorem}
\newtheorem{lemma}[theorem]{Lemma}

\theoremstyle{definition}
\newtheorem{definition}{Definition}
\newtheorem{assumption}{Assumption}
\newtheorem{problem}{Problem}

\theoremstyle{remark}
\newtheorem{remark}{Remark}

\makeatletter
\newcommand{\objitem}[2]{
    \item[\textbf{#1}]
    \edef\@currentlabel{#1}
    \label{#2}
}
\makeatother

\usepackage{cite}
\makeatletter
\let\NAT@parse\undefined
\makeatother
\usepackage[colorlinks=true,
            linkcolor=blue,
            citecolor=blue,
            urlcolor=blue,
            pdfborder={0 0 0}]{hyperref}

\hypersetup{
    pdfauthor={Hamza Tariq and Adeel Akhtar},
    pdftitle={Decentralized Safe Path Following for Multiple Quadrotors on Intersecting Paths with Theoretical Guarantees},
    pdfdisplaydoctitle=true,
    pdfsubject={Safety-Critical Control, Path Following},
    pdfkeywords={safe path following, quadrotors, control barrier functions, transverse feedback linearization},
    bookmarksnumbered,
    pdfstartview={FitV}
}

\begin{document}

\title{Decentralized Safe Path Following for Multiple Quadrotors on Intersecting Paths with Theoretical Guarantees}

\author{Hamza Tariq$^{1}$ and Adeel Akhtar$^{1}$
\thanks{$^{1}$Hamza Tariq and Adeel Akhtar are with the Department of Mechanical \& Industrial Engineering at the New Jersey Institute of Technology, Newark, NJ 07102, USA. {\tt\small \{ht299, adeel.akhtar\}@njit.edu}}}

\maketitle

\begin{abstract}
This paper studies decentralized safe path following for multiple quadrotors on intersecting paths, where safety requires both collision avoidance and strict adherence to pre-assigned routes. The proposed controller reformulates transverse feedback linearization as a constrained quadratic program with four equality constraints: two enforce convergence to and strict adherence to the path, and two prescribe the desired speed and heading. Safety is achieved by relaxing only the along-path speed constraint, while the path and heading constraints remain hard. Under the stated assumptions, and for admissible initial conditions, the controller guarantees that all agents converge to and thereafter follow their assigned paths, avoid collisions, and avoid attitude singularities. We evaluate the controller in the Drake physics engine on non-planar intersecting paths and compare it against two nominal-plus-safety-filter cascades. Code and additional results are available at {\footnotesize \url{https://gradslab.github.io/safe_multiquad_pf/}}.
\end{abstract}

\section{INTRODUCTION}

Multi-quadrotor systems are increasingly deployed in applications that require multiple vehicles to follow prescribed routes within a shared airspace, such as warehouse inventory and cooperative inspection of power lines and bridges. A vehicle that deviates from its route is both a collision hazard and a worse sensor, because what it records depends on where it flies~\cite{beul2018fast,silano2021power}. We study \emph{safe path following}, which requires simultaneously enforcing collision avoidance and strict adherence to assigned routes, i.e., \emph{path invariance}.

Path invariance means that once the state of a vehicle reaches its path-following manifold~\cite{nielsen2008local}, formalized in Section~\ref{sec:prelim}, its deviation from that path remains zero.  Path invariance matters most in settings where lateral deviations are not acceptable, for example flight through narrow corridors and tunnels, or inspection at close range. There, the vehicle must also hold a prescribed heading so that an onboard camera stays pointed at the structure~\cite{mansouri2018cooperative}.

Prior work on safe path following for multi-quadrotor systems handles collision avoidance indirectly through nonintersecting path design or timing coordination~\cite{Cichella2012cdc,Mehdi2017cdc,xie2020cooperative}. These strategies are not suitable in confined environments, where nonintersecting paths might not exist and agents repeatedly pass through shared intersections. An alternative is to use control barrier functions (CBFs) or exponential CBFs (ECBFs) to design \emph{safety filters} for collision avoidance~\cite{ames2019control, nguyen2016exponential, wang2017safety}, but these filters generally do not preserve path invariance. For example,~\cite{yao2022guiding} applies a safety filter to a multi-agent guiding-vector-field controller, but the resulting closed-loop system does not guarantee path invariance when paths intersect. Consequently, most work on safe multi-quadrotor systems has focused on trajectory tracking, where deviation from a given route is allowed~\cite{wang2017safe,goarin2025decentralized}. A virtual reference governor has achieved safe path following with obstacle avoidance, but only for a single agent~\cite{li2020fast}. To the best of our knowledge, no decentralized safe path-following controller has been developed that guarantees both collision avoidance and path invariance for multiple quadrotors on intersecting~paths.

\begin{figure}[t]
    \centering
    \includegraphics[width=0.9\linewidth]{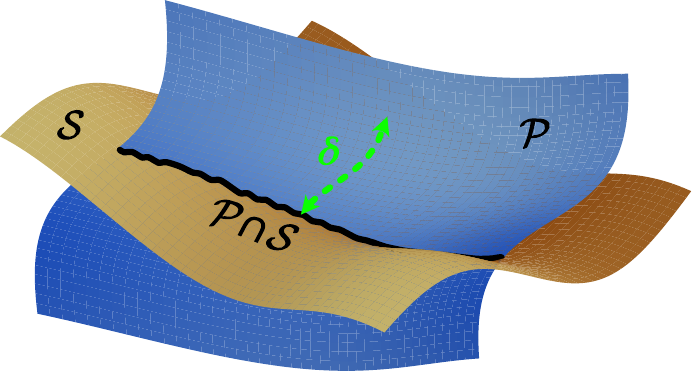}

    \vspace{-0.075in}

    \caption{Manifold interpretation of the proposed relaxation. The
    nominal TFL controller stabilizes the intersection of the
    path-following and secondary-objective manifolds. The scalar slack
    $\delta$ relaxes the along-path speed while the transverse and
    heading equalities remain hard, so the
    path-following manifold remains locally exponentially attractive along admissible closed-loop trajectories.
        }
    \label{fig:manifold_intersection_concept}

    \vspace{-0.25in}

\end{figure}

Transverse feedback linearization (TFL) provides a natural framework for the safe path-following problem~\cite{NieMag04,AkhWasNie2013Journal,akhtar2012path}.  Specifically, TFL splits the control task into a primary objective, which drives the quadrotor to the desired path and keeps it there, and a secondary objective, which determines its speed and heading along the path~\cite{roza2012path}. For the dynamically extended quadrotor model, this separation gives four decoupled integrator chains. Stabilizing two of them stabilizes the path-following manifold, on which the quadrotor stays on its path, and stabilizing the other two stabilizes the secondary-objective manifold, where it maintains the desired speed and heading (Fig.~\ref{fig:manifold_intersection_concept}).
 We first reformulate the TFL controller in~\cite{akhtar2012path} as an equivalent parametric quadratic program (QP) with four equality constraints: two \emph{transversal} constraints that stabilize the path-following manifold and two \emph{tangential} constraints that stabilize the secondary-objective manifold.
We then augment the QP with higher-order ECBF constraints~\cite{nguyen2016exponential} to ensure that each quadrotor avoids attitude singularities and collisions with other agents, extending the decentralization approach in~\cite{wang2017safety} to safe path following. This four-constraint structure is the central strength of the framework. Each equality constraint realizes exactly one control objective, so a designer may relax any subset of the four, including all four, by attaching a slack variable to the chosen constraints. Relaxing all four, while minimizing the deviation from the nominal input, recovers the standard nominal-plus-safety-filter cascade (Remark~\ref{rem:cascade}). The choice sets a tradeoff between guarantees and feasibility, because every constraint kept hard is an objective enforced exactly, yet each hard constraint also shrinks the feasible set of the QP. Here we treat three constraints as hard and relax only the speed, so the speed channel is the only nominal objective allowed to deviate, while the nominal transverse and heading-error dynamics are preserved exactly. The main contributions of this paper are as follows:
\begin{itemize}
    \item We propose the first decentralized path-following controller for multiple quadrotors on intersecting paths that guarantees path invariance, collision avoidance, and singularity avoidance under the stated assumptions.

    \item We reformulate the TFL controller as a QP that relaxes only the speed constraint, preserves the transverse and heading dynamics, and admits a closed-form solution (Remark~\ref{rem:interval}).

    \item We validate the controller in the Drake physics engine~\cite{drake} on non-planar intersecting paths against two safety-filtered cascade baselines.
\end{itemize}

\textit{Notation:} For scalars \(x_1,\dots,x_n\), \(\mathbf x:=\operatorname{col}(x_1,\dots,x_n)\in\mathbb R^n\). The symbols \(\|\cdot\|\) and \(\|\mathbf x\|_G:=\sqrt{\mathbf x^\top G\mathbf x}\), with \(G\succ0\), denote the Euclidean and weighted norms, respectively. The operators \(\operatorname{diag}(\cdot)\) and \(\operatorname{blkdiag}(\cdot)\) form diagonal and block-diagonal matrices, and \(\prod\) denotes the Cartesian product. For a smooth scalar function \(\lambda\) and smooth vector fields \(f\) and \(g\), define \(L_f^0\lambda:=\lambda\), \(L_f^k\lambda:=L_f(L_f^{k-1}\lambda)\) for \(k\ge1\), and \(L_gL_f^k\lambda:=L_g(L_f^k\lambda)\).

\section{Preliminaries}\label{sec:prelim}

\subsection{Exponential Control Barrier Functions}

Consider a nonlinear control-affine system with state $\mathbf x\in\mathbb R^n$, control input $\mathbf u \in \mathbb R^m$, and dynamics
\begin{equation}\label{eq:generic_system}
    \dot{\mathbf x} = f(\mathbf x) + g(\mathbf x)\mathbf u,
\end{equation}
where $f:\mathbb R^n \to \mathbb R^n$ and $g: \mathbb R^n \to \mathbb R^{n \times m}$ are locally Lipschitz continuous. Given a smooth function $h:\mathbb R^n\to\mathbb R$, define the safe set $\mathcal C\eqdef\{\mathbf x\in\mathbb R^n:h(\mathbf x)\ge0\}$. The function $h$ is a CBF if it satisfies the standard relative-degree-one CBF condition in~\cite[Definition 2]{ames2019control}, and any locally Lipschitz controller enforcing this condition renders $\mathcal C$ forward invariant. When $L_gh\equiv\mathbf 0$, $h$ has relative degree greater than one, which motivates \emph{exponential} CBFs.

\begin{definition}[Exponential Control Barrier Function~\cite{nguyen2016exponential}]\label{def:ecbf}
    A smooth function $h:\mathbb R^n\to \mathbb R$, with relative degree $r\in \mathbb N$, defining the set $\mathcal{C}$ is an Exponential CBF (ECBF) for~\eqref{eq:generic_system} if there exists $\bm k \in \mathbb R^r$ such that the polynomial $s^r+k_r s^{r-1}+\dots+k_2 s+k_1=\prod_{i=1}^r (s-\nu_i)$ has negative real roots $(\nu_1,\dots,\nu_r) \in \mathbb R^r$ and for all $\mathbf x\in\mathbb R^n$,
        $\sup_{\mathbf u\in\mathbb R^m}\{L_f^rh(\mathbf x) +L_gL_f^{r-1}h(\mathbf x)\mathbf u\} \ge -\bm k^\top H(\mathbf x)$,
    where $H(\mathbf x)\eqdef \operatorname{col}(h(\mathbf x),\,L_f h(\mathbf x),\,\dots,\,L_f^{r-1}h(\mathbf x))$.
\end{definition}
Now, for a given ECBF, we define the functions $\Lambda_0(\mathbf x) \eqdef h(\mathbf x)$ and $\Lambda_i(\mathbf x) \eqdef d\Lambda_{i-1}/dt - \nu_i\Lambda_{i-1}$ for all $i\in\{1,\dots,r-1\}$.
The corresponding ECBF safe set is then given by
\begin{equation}\label{eq:ecbf_safe_sets}
    \mathcal B \eqdef \cap_{i=0}^{r-1}\mathcal C_i \eqdef \cap_{i=0}^{r-1}\{\mathbf x \in \mathbb R^n:\Lambda_i(\mathbf x)\ge0\}.
\end{equation}
We now review the main result on ECBFs from~\cite{nguyen2016exponential}.
\begin{theorem}\label{thm:ecbf}
    If $h$ is an ECBF for~\eqref{eq:generic_system}, then any locally Lipschitz feedback controller $\kappa: \mathbb R^n\to\mathbb R^m$ satisfying the constraint $L_f^rh(\mathbf x) +L_gL_f^{r-1}h(\mathbf x)\kappa(\mathbf x)\ge -\bm k^\top H(\mathbf x)$ for all $\mathbf x\in \mathbb R^n$ renders $\mathcal B$~\eqref{eq:ecbf_safe_sets} forward invariant under the closed-loop dynamics induced by $\mathbf u=\kappa(\mathbf x)$. Hence, if $\mathbf x_0 \in \mathcal B$, then $h(\mathbf x(t))\ge0$ for all $t\ge0$.
\end{theorem}

ECBFs are commonly deployed as \emph{safety filters}~\cite{ames2019control,nguyen2016exponential}: given a nominal input $\mathbf u_{\mathrm{nom}}=\kappa_{\mathrm{nom}}(\mathbf x)$, the applied input is
\begin{equation}\label{eq:qp_safety_filter}
\begin{aligned}
\kappa(\mathbf x)=\argmin_{\mathbf u\in\mathbb R^m}\;
& \tfrac12\|\mathbf u-\mathbf u_{\mathrm{nom}}\|^2_{W}\\
\text{s.t.}\;
& L_f^rh(\mathbf x)+L_gL_f^{r-1}h(\mathbf x)\mathbf u \ge -\bm k^\top H(\mathbf x),
\end{aligned}
\end{equation}
with $W\succ0$.

\subsection{Multi-Quadrotor System Model}
We consider $N$ homogeneous quadrotors indexed by $\mathcal V\eqdef\{1,\dots,N\}$ over a complete communication graph, so each agent can transmit its state to every other agent. Each agent computes its own control input locally rather than through a centralized optimizer, which is the sense of decentralized in~\cite{wang2017safety,goarin2025decentralized}. For agent $i\in\mathcal V$, let $x_q^i\eqdef\col(x^i,y^i,z^i)\in\mathbb R^3$ and $v_q^i\eqdef\col(v_x^i,v_y^i,v_z^i)$ be the position and velocity, $\Phi^i\eqdef\col(\phi^i,\theta^i,\psi^i)$ the roll, pitch, and yaw angles, and $\omega^i\eqdef\col(p^i,q^i,r^i)$ the body rates. With inertia $J\eqdef\operatorname{diag}(I_x,I_y,I_z)$, mass $m$, gravity $g$, $b_3\eqdef\col(0,0,1)$, torques $\tau^i\in\mathbb R^3$, and collective thrust $u_f^i\in\mathbb R$, the dynamics are~\cite{bouabdallah2005backstepping}
\begin{equation}\label{eq:quad_dynamics}
\begin{aligned}
    \dot\Phi^i &= Y^{i}\omega^i, & J\dot\omega^i &= \tau^i-(\omega^i\times J\omega^i),\\
    \dot x_q^i &= v_q^i, & m\dot v_q^i &= mgb_3-u_f^iR^ib_3,
\end{aligned}
\end{equation}
where $Y^i$ and $R^i$ are standard Euler-angle kinematic and rotation matrices~\cite{roza2012path}. We take the output $\mathbf y^i \in \mathbb R^3$ to be the inertial position $\mathbf y^i=x_q^i$.

\subsection{Dynamic Transverse Feedback Linearization (TFL)}

We now define the path-following problem for a single agent, adapted from \cite{roza2012path}.

\begin{definition}[Path-following problem]\label{def:pf}
Consider a quadrotor~\eqref{eq:quad_dynamics} and a prescribed smooth path
$\gamma\subset\mathbb{R}^3$. Design a feedback controller such that the closed-loop system satisfies the following primary and secondary objectives:
\begin{enumerate}[leftmargin=2.5em]
    \objitem{PO}{pf_po}
    the output converges to the path $\gamma$ and remains on it for all future time;
    \objitem{SO}{pf_so}
    the path-coordinate rate (speed) and yaw angle converge to prescribed values along the path $\gamma$.
\end{enumerate}
\end{definition}
As in~\cite{akhtar2012path}, we apply dynamic extension to~\eqref{eq:quad_dynamics} and render the system suitable for path following via TFL. In particular, a well-defined vector relative degree can be obtained by appending two integrators to the input $u_f^i$. Let $x_{13}^i \eqdef u_f^i$, $x_{14}^i \eqdef \dot u_f^i$, and $u_d^i \eqdef \dot x_{14}^i$. The resulting extended dynamics for agent $i$ are \eqref{eq:quad_dynamics}, with the thrust $u_f^i=x_{13}^i$ now a state, together with
\begin{equation}\label{eq:ext_quad_dynamics}
\dot x_{13}^i = x_{14}^i, \qquad
\dot x_{14}^i = u_d^i.
\end{equation}
With a slight abuse of notation, we define the extended state as $\mathbf x^i \eqdef \operatorname{col}(\Phi^i,\,\omega^i,\,x_q^i,\,v_q^i,x_{13}^i,x_{14}^i) \in \mathbb R^{14}$ and the extended
control input as $\mathbf u^i=\operatorname{col}(u_d^i,\tau_p^i,\tau_q^i,\tau_r^i) \in \mathbb R^4$. The output is still the inertial position $\mathbf y^i = c(\mathbf x^i) = x_q^i$, where $c:\mathbb R^{14}\to\mathbb R^3$. The extended model, \eqref{eq:quad_dynamics} together with~\eqref{eq:ext_quad_dynamics}, can be written in control-affine form
\begin{equation}\label{eq:extended_control_affine}
    \dot{\mathbf x}^i = f(\mathbf x^i) + g(\mathbf x^i)\mathbf u^i,
\end{equation}
for some smooth maps $f$ and $g$. For details, see~\cite{akhtar2012path}. Next, we define the combined state and input of the multi-quadrotor system as $\mathbf x\eqdef \operatorname{col}(\mathbf{x}^1,\dots,\mathbf x^N) \in \mathbb R^{14N}$ and $\mathbf u:=\operatorname{col}(\mathbf u^1,\dots,\mathbf u^N)\in\mathbb{R}^{4N},$ respectively. Moreover, we define the combined closed-loop dynamics as
\begin{equation}\label{eq:full_dynamics}
    \dot{\mathbf x}= F(\mathbf x)+G(\mathbf x)\mathbf u,
\end{equation}
where
$F(\mathbf x):=\operatorname{col}(f(\mathbf x^1),\dots,f(\mathbf x^N)),$ and
$G(\mathbf x):=\operatorname{blkdiag}(g(\mathbf x^1),\dots,g(\mathbf x^N))$.

As in~\cite{akhtar2012path}, for each agent \(i\), let
\(\gamma^i\subset\mathbb R^3\) be a smooth embedded regular curve with a
smooth regular parameterization
\(\sigma^i:\operatorname{dom}(\sigma^i)\to\mathbb R^3\), where
\(\operatorname{dom}(\sigma^i)\subset\mathbb R\) is an interval. For a closed path, \(\operatorname{dom}(\sigma^i)=\mathbb R\) and \(\sigma^i\) is periodic with parameter period \(T_i>0\).
The parameter defines an along-path coordinate whose desired rate
is specified through the secondary objective. Throughout, we refer
to this path-coordinate rate as the along-path speed. Let \(s^i:\mathbb R^3\to\mathbb R^2\), with \(s^i=\operatorname{col}(s_1^i,s_2^i)\), be a smooth map such that
\(\gamma^i=\{\mathbf y^i\in\mathbb R^3:s^i(\mathbf y^i)=\bm 0\}\) and \(\operatorname{rank}Ds^i(\mathbf y^i)=2\) for all \(\mathbf y^i\in\gamma^i\),
where \(Ds^i(\mathbf y^i)\in\mathbb R^{2\times3}\) denotes the Jacobian of \(s^i\). Thus, the two components of \(s^i\) provide independent local transverse coordinates to the path. Define \(\alpha^i\eqdef\operatorname{col}(\alpha_1^i,\alpha_2^i)\), where \(\alpha_1^i\eqdef s_1^i\circ c\) and \(\alpha_2^i\eqdef s_2^i\circ c\) are smooth maps from \(\mathbb R^{14}\) to \(\mathbb R\). We lift \(\gamma^i\) to the state space of the extended model~\eqref{eq:extended_control_affine} by defining \(\Gamma^i\eqdef\{\mathbf x^i\in\mathbb R^{14}:\alpha_1^i(\mathbf x^i)=0,\ \alpha_2^i(\mathbf x^i)=0\}\). Thus, output convergence to \(\gamma^i\), corresponding to the primary objective~\ref{pf_po}, is equivalent to state convergence to \(\Gamma^i\). In general, \(\Gamma^i\) is not controlled invariant. We therefore seek to stabilize its maximal controlled-invariant subset~\cite{nielsen2008local}
\begin{equation}\label{eq:path_following_manifold}
    \mathcal P^i=
    \{\mathbf x^i\in\mathbb{R}^{14}: L_f^k\alpha^i(\mathbf x^i)=0,\ k=0,1,2,3\},
\end{equation}
and we call it the \emph{path-following manifold} for agent $i$. We then define the combined path-following manifold as the Cartesian product \(\mathcal P \eqdef \prod_{i\in\mathcal V} \mathcal P^i \subset \mathbb R^{14N}\).
Since the system has four inputs and $\alpha^i$ has two components, we define two additional maps, $\beta_1^i$ and $\beta_2^i$, for each agent $i$. Let \(U_\gamma^i\subset\mathbb R^3\) be a tubular neighborhood of
\(\gamma^i\) on which the nearest-point path coordinate is unique,
and define the smooth map
\(\varpi^i:U_\gamma^i\to\operatorname{dom}(\sigma^i)\) by
\(\varpi^i(\mathbf y^i)\eqdef
\operatorname*{arg\,min}_{d\in\operatorname{dom}(\sigma^i)}
\|\mathbf y^i-\sigma^i(d)\|\). For non-compact paths, a tubular
neighborhood of uniform size exists whenever the path is uniformly
embedded and of bounded
geometry~\cite[Thm.~2.31]{eldering2013noncompact}. We then define
\(\beta_1^i\eqdef\varpi^i\circ c\), which gives the path coordinate
of agent \(i\). This nearest-point construction is one canonical choice. More generally, \(\beta_1^i\) may be any smooth function that restricts on \(\Gamma^i\) to the path parameter of the output and satisfies the relative-degree condition imposed on \(U^i\) below, because the constructions that follow use only these two properties. The simulations in Section~\ref{sec:sim} use such a coordinate. For a closed path, the path coordinate \(\beta_1^i\) is interpreted in a local path-coordinate chart. Since the controller depends only on the Lie derivatives \(L_f^j\beta_1^i\), \(j\ge1\), which are unchanged under a change of local lift by an additive constant, the construction below is independent of the particular local representation. On each connected component of a chart overlap, the two representations differ by an additive constant, so the controller takes the same value in either chart. Next, we define $\beta_2^i(\mathbf x^i)\eqdef\psi^i$, the yaw angle of agent $i$, and the virtual output
\begin{equation}\label{eq:virtual_output}
    \hat{\mathbf y}^i
    \eqdef
    \operatorname{col}\big(\alpha_1^i(\mathbf x^i),\,\alpha_2^i(\mathbf x^i),\,\beta_1^i(\mathbf x^i),\,\beta_2^i(\mathbf x^i)\big).
\end{equation}
Define
\(\mathcal X_D^i\eqdef
\{\mathbf x^i\in\mathbb R^{14}:
\phi^i\neq\pm\pi/2,\ \theta^i\neq\pm\pi/2,\ x_{13}^i\neq0\}\). For
maps \(s^i\) in which \(s_1^i\) is independent of altitude and
\(s_2^i\) depends on altitude alone, \cite[Lemma
5.1]{akhtar2012path} shows that the virtual
output~\eqref{eq:virtual_output} has a well-defined vector relative
degree \(\{4,4,4,2\}\) at every point of
\(\mathcal P^i\cap\mathcal X_D^i\). We make no such claim for a
general path representation. Instead, we restrict the operating set
so that the relative-degree condition holds on it.
For each agent, fix an open set
\(U^i\subset\mathcal X_D^i\cap c^{-1}(U_\gamma^i)\), containing $\mathcal P^i\cap\mathcal X_D^i$, on which the virtual output~\eqref{eq:virtual_output} has vector relative degree \(\{4,4,4,2\}\) and the mapping
\(T^i:U^i\to\mathbb R^{14}\), defined by
\begin{equation}\label{eq:transform}
T^i(\mathbf x^i)
\eqdef
\operatorname{col}\!\left(
(\xi_j^i)_{j=1}^{4},
(\zeta_j^i)_{j=1}^{4},
(\eta_j^i)_{j=1}^{4},
(\mu_k^i)_{k=1}^{2}
\right),
\end{equation}
with \(\xi_j^i\eqdef L_f^{j-1}\alpha_1^i\), \(\zeta_j^i\eqdef L_f^{j-1}\alpha_2^i\), and \(\eta_j^i\eqdef L_f^{j-1}\beta_1^i\) for \(j\in\{1,\ldots,4\}\), and \(\mu_k^i\eqdef L_f^{k-1}\beta_2^i\) for \(k\in\{1,2\}\),
is a local diffeomorphism at every point and the decoupling matrix
\(D^i\) in~\eqref{eq:decoupling_matrix} below is nonsingular on \(U^i\). For the specialized maps of~\cite[Lemma 5.1]{akhtar2012path}, such a set exists because the vector relative degree makes \(D^i\) invertible at every point in $\mathcal P^i\cap\mathcal X_D^i$, hence on an open neighborhood of it by continuity. For the path representations used in Section~\ref{sec:sim}, the relative-degree conditions are confirmed by direct computation of the Lie derivatives of the virtual output. Since
\(\det D^i\) depends only on the attitude, the thrust, and the
position, membership in \(U^i\) constrains no other coordinate. For closed paths, the transformation \(T^i\) is understood locally in a path-coordinate chart, and a local diffeomorphism is sufficient for the analysis~below.

Define
\(\Omega\eqdef\prod_{i\in\mathcal V}U^i\) and
\(\mathcal P_\Omega\eqdef\mathcal P\cap\Omega\). Thus,
\(\Omega\subset\mathbb R^{14N}\) is an open operating neighborhood
of \(\mathcal P_\Omega\), and the
coordinate transformation~\eqref{eq:transform} is valid for every
agent \(i\in\mathcal V\) and every
\(\mathbf x=\operatorname{col}(\mathbf x^1,\ldots,\mathbf x^N)
\in\Omega\). Consequently, using a feedback transformation similar
to that in~\cite{akhtar2012path}, the
extended model~\eqref{eq:extended_control_affine} for each agent \(i\) is reduced
to four decoupled chains of integrators, the \(\xi\)-, \(\zeta\)-, \(\eta\)-, and \(\mu\)-subsystems,
\begin{equation}\label{eq:linear_systems}
\begin{aligned}
&\dot\xi_j^i=\xi_{j+1}^i,\ \ \dot\zeta_j^i=\zeta_{j+1}^i,\ \ \dot\eta_j^i=\eta_{j+1}^i,\quad j\in\{1,2,3\},\\
&\dot\xi_4^i=v_\xi^i,\ \ \dot\zeta_4^i=v_\zeta^i,\ \ \dot\eta_4^i=v_\eta^i,\ \ \dot\mu_1^i=\mu_2^i,\ \ \dot\mu_2^i=v_\mu^i.
\end{aligned}
\end{equation}
Next, the origin of the $\xi$- and $\zeta$-subsystems is stabilized using the auxiliary \emph{transversal} control inputs
\begin{equation}\label{eq:transversal_in}
v_{\xi}^{i} = -\textstyle\sum_{j=1}^{4} k_{\xi_j}^{i}\,\xi_j^{i},
\qquad
v_{\zeta}^{i} = -\textstyle\sum_{j=1}^{4} k_{\zeta_j}^{i}\,\zeta_j^{i},
\end{equation}
 where $k^i_{\xi_j},k^i_{\zeta_j}>0$ for $j\in\{1,\dots,4\}$ are chosen so that the corresponding characteristic polynomials are Hurwitz~\cite{akhtar2012path}. The transversal inputs~\eqref{eq:transversal_in} locally stabilize
the manifold $\mathcal P^i\cap U^i$ relative to $U^i$,
thereby achieving the primary objective~\ref{pf_po}.

Note that \(\eta_2^i\) is the path-coordinate rate variable, while
\(\mu_1^i\) is the yaw variable. Given a constant desired speed \(v^i_{\text{des}}\) and constant desired yaw
\(\psi^i_{\text{des}}\), which may differ between agents, let
$\beta_v^i\eqdef L_f\beta_1^i(\mathbf x^i)-v^i_{\text{des}}$
and
$\beta_\psi^i\eqdef \beta_2^i(\mathbf x^i)-\psi^i_{\text{des}}$. We now define the secondary-objective manifold as
\(\mathcal S^i \eqdef \{\mathbf x^i\in\mathbb{R}^{14}:\beta_v^i=L_f^2\beta_1^i(\mathbf x^i)=L_f^3\beta_1^i(\mathbf x^i)=0,\ \beta_\psi^i=L_f\beta_2^i(\mathbf x^i)=0\}\).
Stabilizing \(\mathcal S^i\) is equivalent to achieving~\ref{pf_so}. To do this, we choose the auxiliary \emph{tangential} control inputs
\begin{equation}\label{eq:tang_in}
    \begin{aligned}
        v_\eta^i &
=
-k_{\eta_2}^i(\eta_2^i-v^i_{\text{des}})
-k_{\eta_3}^i\eta_3^i
-k_{\eta_4}^i\eta_4^i,\\
v_\mu^i&
=
-k_{\mu_1}^i(\mu_1^i-\psi^i_{\text{des}})
-k_{\mu_2}^i\mu_2^i,
    \end{aligned}
\end{equation}
where the gains are chosen so that the two characteristic polynomials are Hurwitz. Smooth time-varying references can be accommodated by adding the corresponding feedforward terms to~\eqref{eq:tang_in}, but constant references suffice here. Combining the transversal and tangential inputs gives the TFL controller
\begin{equation}\label{eq:tfl_controller}
    \mathbf u_{\mathrm{TFL}}^i(\mathbf x^i) = \big(D^i(\mathbf x^i)\big)^{-1}\mathbf v^i(\mathbf x^i),
\end{equation}
where \(D^i:\mathbb R^{14} \to \mathbb R^{4\times 4}\) and \(\mathbf v^i:\mathbb R^{14} \to \mathbb R^{4}\) are defined by
\begin{equation}\label{eq:decoupling_matrix}
D^i(\mathbf x^i):=
\begin{bmatrix}
L_gL_f^3\alpha_1^i(\mathbf x^i)\\
L_gL_f^3\alpha_2^i(\mathbf x^i)\\
L_gL_f^3\beta_1^i(\mathbf x^i)\\
L_gL_f\beta_2^i(\mathbf x^i)
\end{bmatrix},
\;
\mathbf v^i(\mathbf x^i):=
\begin{bmatrix}
-L_f^4\alpha_1^i(\mathbf x^i)+v_\xi^i\\
-L_f^4\alpha_2^i(\mathbf x^i)+v_\zeta^i\\
-L_f^4\beta_1^i(\mathbf x^i)+v_\eta^i\\
-L_f^2\beta_2^i(\mathbf x^i)+v_\mu^i
\end{bmatrix},
\end{equation}
respectively. Therefore, for each agent \(i\in\mathcal V\), the TFL
controller~\eqref{eq:tfl_controller} locally solves the
path-following problem in Definition~\ref{def:pf} by stabilizing
\((\mathcal P^i\cap\mathcal S^i)\cap U^i\), provided that the decoupling matrix \(D^i(\mathbf x^i)\) in~\eqref{eq:decoupling_matrix} remains invertible along the closed-loop trajectories. We refer to the roll and pitch angles \(\phi^i = \pm\pi/2\) and \(\theta^i = \pm\pi/2\), where the TFL controller~\eqref{eq:tfl_controller} is undefined, as \emph{attitude singularities}.

\section{Problem Formulation}\label{sec:prob}

A common way to handle multi-agent collision avoidance is to use CBFs or ECBFs to design a \emph{safety filter}~\eqref{eq:qp_safety_filter}, which keeps each agent $i\in\mathcal V$ at least a safe distance $d_s>0$ from every other agent $j\in\mathcal V\setminus\{i\}$~\cite{wang2017safety,wang2017safe,goarin2025decentralized}.
However, when the prescribed paths intersect, cascading the nominal TFL controller~\eqref{eq:tfl_controller} with the safety filter~\eqref{eq:qp_safety_filter} does not generally guarantee path invariance. Because the filter chooses the smallest control change that enforces the collision-avoidance constraint, agent $i$ may deviate slightly from its path instead of briefly stopping on it. The filter also perturbs the heading, which turns any body-fixed sensor away from its target.

The goal of this paper is therefore to design, for each agent, a decentralized feedback controller that guarantees collision avoidance while achieving the primary objective~\ref{pf_po}. Since the TFL controller is well defined only away from attitude singularities, the controller must also ensure that the closed-loop trajectories avoid them for all time.

\begin{assumption}\label{ass:non_zero_thrust}
For all $t \ge 0$, the collective thrust for each agent $i \in \mathcal V$ in the multi-quadrotor system satisfies $x^i_{13} > 0$.
\end{assumption}

Since $x_{13}^i=u_f^i$ is the total rotor thrust, Assumption~\ref{ass:non_zero_thrust} excludes only unusual operating conditions, such as free fall and aggressive acrobatic maneuvers. It can also be enforced by one more ECBF row for \(x_{13}^i-f_{\min}\), \(f_{\min}>0\), which has relative degree~$2$ in \(u_d^i\).

\begin{problem}\label{prob:safe_pf}
Consider a multi-quadrotor system of $N$ homogeneous quadrotors with safe distance $d_s>0$. For each agent $i\in\mathcal V$, let the dynamics be given by the extended model~\eqref{eq:extended_control_affine} under Assumption~\ref{ass:non_zero_thrust}, with path-following manifold $\mathcal P^i$ defined in~\eqref{eq:path_following_manifold} and virtual output $\hat{\mathbf y}^i$ defined in~\eqref{eq:virtual_output}. Design a decentralized controller
$\kappa^i: \mathbb R^{14N} \to \mathbb R^4$
such that, for all admissible initial conditions in a neighborhood
of $\mathcal P_\Omega$ relative to $\Omega$, where the admissible
set is made precise in Section~\ref{sec:controller}, the
closed-loop system
\begin{enumerate}
    \item achieves~\ref{pf_po} by rendering $\mathcal P^i\cap U^i$ locally exponentially attractive along admissible closed-loop trajectories,
    \item avoids attitude singularities \(\phi^i=\pm \pi/2,\ \theta^i=\pm \pi/2\) with safety margin \(\epsilon>0\) for all $t\ge0$, and
    \item guarantees pairwise collision avoidance, that is,
    $\|x_q^i(t)-x_q^j(t)\| \ge d_s$ for all $\{i,j\}\subset\mathcal V$ and all $t\ge0$.
\end{enumerate}
\end{problem}

\section{Proposed Controller}\label{sec:controller}
To incorporate safety, the nominal TFL controller~\eqref{eq:tfl_controller} must be relaxed, but only in part: we relax its speed component alone. To enable this \emph{selective} relaxation, we reformulate the TFL controller as a QP. In particular, \eqref{eq:tfl_controller} can be written as \(D^i(\mathbf x^i)\mathbf u^i_{\mathrm{TFL}}=\mathbf v^i(\mathbf x^i)\), which is affine in \(\mathbf u^i_{\mathrm{TFL}}\) because the inputs~\eqref{eq:transversal_in} and~\eqref{eq:tang_in} depend only on \(\mathbf x^i\).  Hence, for each \(i\in\mathcal V\), the TFL controller~\eqref{eq:tfl_controller} can be written equivalently as the parametric QP
\begin{subequations}\label{eq:tfl_qp}
\begin{align}
\kappa^i_{\mathrm{TFL}}(\mathbf x^i) = \argmin_{\mathbf u^i\in\mathbb R^4}\quad
& \frac{1}{2}\|\mathbf u^i\|^2 \\
\text{s.t.}\quad
& D^i(\mathbf x^i)\mathbf u^i=\mathbf v^i(\mathbf x^i) \label{eq:tfl_constraint_hard}.
\end{align}
\end{subequations}
Since \(D^i\) in~\eqref{eq:decoupling_matrix} is invertible on \(U^i\) by construction (Section~\ref{sec:prelim}), the constraint~\eqref{eq:tfl_constraint_hard} has the unique solution \(\mathbf u^i=(D^i)^{-1}\mathbf v^i\), which is therefore the unique optimizer, so on \(U^i\) the TFL-QP~\eqref{eq:tfl_qp} coincides with the TFL controller~\eqref{eq:tfl_controller}.

Therefore, we can now add safety constraints to the nominal controller itself.  We keep the two transversal rows hard so that, whenever the QP is feasible, the nominal transverse dynamics are preserved exactly, and the heading row hard to preserve the pointing direction of body-fixed sensors.

\subsection{Roll and Pitch Constraints}

Let $a_\epsilon\eqdef \tfrac{\pi}{2}-\epsilon$, where $\epsilon>0$ is the prescribed singularity margin. For each agent $i\in\mathcal V$ and angle $\rho\in\{\phi,\theta\}$, define
\begin{equation}\label{eq:att_cbfs}
\begin{aligned}
h_{\rho}^{i\pm}(\mathbf x^i)&\eqdef a_\epsilon\mp\rho^i,\\
\mathcal C_{\rho}^{i\pm}&\eqdef
\{\mathbf x^i\in\mathbb R^{14}:h_{\rho}^{i\pm}(\mathbf x^i)\ge0\}.
\end{aligned}
\end{equation}
The $+$ and $-$ barriers enforce the upper and lower bounds on $\rho^i$, respectively. Each of the four barrier functions has relative degree $2$ on \(\mathcal X_D^i\), by direct computation. With $H_{\rho}^{i\pm}\eqdef\operatorname{col}(h_{\rho}^{i\pm},L_fh_{\rho}^{i\pm})$ and gains $\bm k_{\rho}^{i\pm}$ chosen according to Definition~\ref{def:ecbf}, the ECBF constraint, with arguments suppressed, for each $\rho\in\{\phi,\theta\}$ and each sign is
\begin{equation}\label{eq:attitude_ecbf}
L_gL_fh_{\rho}^{i\pm}\,\mathbf u^i
\ge -L_f^2h_{\rho}^{i\pm}
-(\bm k_{\rho}^{i\pm})^\top H_{\rho}^{i\pm}.
\end{equation}

\subsection{Pairwise Collision Constraints}

 For each unordered pair \(\{i,j\}\subset \mathcal V\), define the projection \(c_{ij}(\mathbf x):=\operatorname{col}(\mathbf x^i,\mathbf x^j)\in\mathbb{R}^{28}\) and pairwise collision barrier as
\begin{equation}\label{eq:Bij_clean}
\bar h_{ij}(\mathbf x)
:= h_{ij}(c_{ij}(\mathbf x))
:= \|x_q^i-x_q^j\|^2-d_s^2,
\end{equation}
for a safety distance $d_s>0$. The pairwise safe set is
\(\bar{\mathcal C}_{ij}:=\{\mathbf x\in\mathbb{R}^{14N}:\bar h_{ij}(\mathbf x)\ge 0\}\).
For notational simplicity, we henceforth suppress the argument of $\bar h_{ij}$. The function \(\bar h_{ij}\) has relative degree $4$ with respect to the combined dynamics~\eqref{eq:full_dynamics} on \(\bar{\mathcal C}_{ij}\cap\Omega\), because \(L_{G_i}L_F^3\bar h_{ij}=2(x_q^i-x_q^j)^\top B^i\), where the input matrix \(B^i\) of the fourth derivative of \(x_q^i\) has rank~$3$ on \(\mathcal X_D^i\), and \(x_q^i\neq x_q^j\) on \(\bar{\mathcal C}_{ij}\). Moreover, for every \(\ell\in \mathcal V\setminus\{i,j\}\),
\(L_{G_\ell}L_F^3\bar h_{ij}=0\), where \(G_\ell(\mathbf x)\) denotes the block-column of \(G(\mathbf x)\) associated with agent \(\ell\). Hence, the centralized ECBF constraint for the pair
\(\{i,j\}\) is
\begin{equation}\label{eq:centralized_col_constraint}
L_F^4\bar h_{ij}
+L_{G_i}L_F^3\bar h_{ij}\,\mathbf u^i
+L_{G_j}L_F^3\bar h_{ij}\,\mathbf u^j
\ge -\bm k_{ij}^\top \bar H_{ij},
\end{equation}
where
\(\bar H_{ij}:=\operatorname{col}\!\big(\bar h_{ij},L_F\bar h_{ij},L_F^2\bar h_{ij},L_F^3\bar h_{ij}\big)\), and \(\bm k_{ij}\) is chosen according to Definition~\ref{def:ecbf}.

Let \(w^{i},w^{j}\in(0,1)\) satisfy
\(w^{i}+w^{j}=1\), where $w^{i}$ and $w^{j}$ are responsibility weights used to split the centralized inequality~\eqref{eq:centralized_col_constraint} into decentralized local
constraints~\cite{wang2017safety}. For homogeneous quadrotors, we choose
\(w^{i}=w^{j}=\frac12\). Accordingly, each
agent \(i\in \mathcal V\) enforces
\begin{equation}\label{eq:collision_ecbf_constraint}
L_{G_i}L_F^3\bar h_{ij}\,\mathbf u^i
\ge
-\tfrac12
\big(
L_F^4\bar h_{ij}+\bm k_{ij}^\top \bar H_{ij}
\big),
\end{equation}
for all $j\in \mathcal V\setminus\{i\}$. The corresponding local constraint enforced by agent \(j\) is obtained
by interchanging \(i\) and \(j\). Summing the two local inequalities
recovers the centralized ECBF condition~\eqref{eq:centralized_col_constraint}. For details, see~\cite{wang2017safety}.

We denote by $\mathcal B_{\rho}^{i\pm}\subset\mathcal C_{\rho}^{i\pm}$ and $\bar{\mathcal B}_{ij}\subset\bar{\mathcal C}_{ij}$ the ECBF admissible sets associated with $h_{\rho}^{i\pm}$ and $\bar h_{ij}$, constructed according to~\eqref{eq:ecbf_safe_sets}. Their explicit form is standard and omitted.

\subsection{Decentralized Multi-Quadrotor Safe Path Following}

We now incorporate the constraints~\eqref{eq:attitude_ecbf} and~\eqref{eq:collision_ecbf_constraint} into~\eqref{eq:tfl_qp} while relaxing only the speed row of~\eqref{eq:tfl_constraint_hard}. Introduce a scalar slack $\delta^i\in\mathbb R$ and the selector $E\eqdef\col(0,0,1,0)$, so that $E\delta^i$ perturbs the $\eta$ row alone. With weights $W_i\succeq0$ and $P_i>0$, the decentralized safe path-following controller of agent $i$ is
\begin{subequations}\label{eq:tfl_safe_pf_qp}
\begin{align}
((\mathbf u^i)^\star,(\delta^i)^\star)
=&
\argmin_{\substack{\mathbf u^i\in\mathbb R^4,\,
                   \delta^i\in\mathbb R}}
\quad
\tfrac{1}{2}\|\mathbf u^i\|^2_{W_i}
+\tfrac{1}{2}P_i(\delta^i)^2
\label{eq:cost}\\
&\text{s.t.}\quad
D^i(\mathbf x^i)\mathbf u^i
=\mathbf v^i(\mathbf x^i)+E\delta^i,
\label{eq:relaxed_tfl_constraint}\\
&\phantom{\text{s.t.}\quad}
A^i(\mathbf x)\mathbf u^i
\ge-\mathbf b^i(\mathbf x).
\label{eq:ineq_const}
\end{align}
\end{subequations}
The inequality~\eqref{eq:ineq_const} collects the \(N+3\) safety constraints of agent \(i\). The first four rows of \(A^i(\mathbf x)\in\mathbb R^{(N+3)\times4}\) are \(L_gL_fh_{\rho}^{i\pm}\), one for each \(\rho\in\{\phi,\theta\}\) and each sign; the remaining \(N-1\) rows are \(L_{G_i}L_F^3\bar h_{ij}\), one for each \(j\in\mathcal V\setminus\{i\}\). The vector \(\mathbf b^i(\mathbf x)\in\mathbb R^{N+3}\) stacks the corresponding terms in~\eqref{eq:attitude_ecbf} and~\eqref{eq:collision_ecbf_constraint}, with responsibility weight \(w^i=\tfrac12\) in each collision row. The attitude rows depend only on \(\mathbf x^i\), whereas the collision rows depend on the neighbors' states.

\subsection{Safe Path Following Guarantees}

By Theorem~\ref{thm:ecbf}, forward invariance requires the feedback induced by~\eqref{eq:tfl_safe_pf_qp} to be locally Lipschitz, a property that is not automatic for parametric QPs~\cite{mestres2025regularity}. Eliminating $\mathbf u^i$ through~\eqref{eq:relaxed_tfl_constraint} with $M^i(\mathbf x)\eqdef(D^i(\mathbf x^i))^{-1}$ gives $\mathbf u^i=M^i(\mathbf v^i+E\delta^i)$ and reduces~\eqref{eq:tfl_safe_pf_qp} to the scalar problem
\begin{subequations}\label{eq:reduced_qp}
\begin{align}
(\delta^i)^\star(\mathbf x)=\argmin_{\delta^i\in\mathbb R}\;
& \tfrac12 Q^i(\mathbf x)(\delta^i)^2+q^i(\mathbf x)\,\delta^i\\
\text{s.t.}\;
& \bar A^i(\mathbf x)\,\delta^i\le\bar{\mathbf b}^i(\mathbf x), \label{eq:reduce_ineq}
\end{align}
\end{subequations}
where $Q^i\eqdef E^\top(M^i)^\top W_iM^iE+P_i>0$, $q^i\eqdef E^\top(M^i)^\top W_iM^i\mathbf v^i$, $\bar A^i\eqdef-A^iM^iE\in\mathbb R^{N+3}$, and $\bar{\mathbf b}^i\eqdef\mathbf b^i+A^iM^i\mathbf v^i$. The applied input is recovered as
\begin{equation}\label{eq:reconstruction_u_star}
(\mathbf u^i)^\star(\mathbf x)=M^i(\mathbf x)\big(\mathbf v^i(\mathbf x^i)+E(\delta^i)^\star(\mathbf x)\big).
\end{equation}

\begin{remark}\label{rem:interval}
Since $\delta^i\in\mathbb R$, each row of~\eqref{eq:reduce_ineq} with a nonzero coefficient bounds $\delta^i$ from above or below, and a row with zero coefficient is feasible if and only if its right-hand side is nonnegative. The reduced QP is thus feasible exactly when every zero-coefficient row is feasible and the largest lower bound does not exceed the smallest upper bound, which one pass through the $N+3$ rows checks in $O(N)$ operations. The optimizer is then $-q^i/Q^i$ clipped to these bounds~\cite[Sec.~8.1.1]{boyd2004convex}, so no numerical QP solver is needed, and the same check can trigger an online fallback controller, whose design is left for future work.
\end{remark}

We define, with the lifted attitude sets
\(\widetilde{\mathcal B}_{\rho}^{i\pm}\eqdef\{\mathbf x\in\mathbb R^{14N}:\mathbf x^i\in\mathcal B_{\rho}^{i\pm}\}\),
\(\bar\Omega\eqdef\Omega\cap\bigcap_{\{i,j\}}\bar{\mathcal B}_{ij}\cap\bigcap_{i,\rho,\pm}\widetilde{\mathcal B}_{\rho}^{i\pm}\)
as the subset of \(\Omega\) on which every barrier condition
holds and make the following assumption.

\begin{assumption}\label{ass:reg_slater}
For each agent \(i\in\mathcal V\) and every \(\bar{\mathbf x}\in\bar\Omega\), there exists \(\hat\delta^i\in\mathbb R\) such that
\(\bar A^{i}(\bar{\mathbf x})\hat\delta^i
<\bar{\mathbf b}^{i}(\bar{\mathbf x})\).
\end{assumption}

Assumption~\ref{ass:reg_slater} requires strict feasibility of the reduced QP~\eqref{eq:reduced_qp} at every point of \(\bar\Omega\), and analogous assumptions are standard for decentralized CBF-QPs~\cite{wang2017safety,mestres2024distributed,goarin2025decentralized}. As in these works, strict feasibility becomes more difficult to maintain as the number of simultaneously interacting agents increases.

\begin{theorem}\label{thm:qp_local_lipschitz}
Suppose that Assumption~\ref{ass:reg_slater} holds. Then, for each agent \(i\in \mathcal V\) and every \(\mathbf x\in\bar\Omega\), the reduced QP~\eqref{eq:reduced_qp} admits a unique minimizer, and both the optimizer map \(\mathbf x\mapsto (\delta^i)^\star(\mathbf x)\) and the controller~\eqref{eq:reconstruction_u_star} are locally Lipschitz continuous on \(\bar\Omega\).
\end{theorem}

\begin{proof}
Since \(\mathbf x^i\in U^i\) for every \(\mathbf x\in\Omega\), and \(D^i\) is invertible on \(U^i\) by the construction of \(U^i\), all the matrix- and vector-valued maps in~\eqref{eq:reduced_qp} are smooth on \(\Omega\), and \(Q^i>0\) on \(\Omega\) since \(W_i\succeq0\) and \(P_i>0\).  By Assumption~\ref{ass:reg_slater}, the feasible set of the reduced QP is nonempty for every \(\mathbf x\in\bar\Omega\). Since this set is closed and the objective is strongly convex, the reduced QP admits a unique minimizer \((\delta^i)^\star(\mathbf x)\). Now, take an arbitrary point \(\bar{\mathbf x}\in\bar\Omega\). On a neighborhood of \(\bar{\mathbf x}\), the reduced QP is a parametric QP with a scalar decision variable, a strongly convex objective, and affine constraints with smooth coefficients, and Assumption~\ref{ass:reg_slater} states exactly Slater's condition for this problem. Although~\cite[Prop.~3.3]{mestres2025regularity} states its hypotheses for data defined on all of \(\mathbb R^{14N}\), extending the QP coefficients smoothly outside this neighborhood changes nothing near \(\bar{\mathbf x}\), so the proposition applies. Consequently, the optimizer map \(\mathbf x\mapsto (\delta^i)^\star(\mathbf x)\) is locally Lipschitz at \(\bar{\mathbf x}\) and, since \(\bar{\mathbf x}\) was arbitrary, on \(\bar\Omega\).  Finally, the controller~\eqref{eq:reconstruction_u_star} is a composition of smooth maps with a locally Lipschitz map, hence \(\mathbf x\mapsto (\mathbf u^i)^\star(\mathbf x)\) is locally Lipschitz on \(\bar\Omega\).
\end{proof}

\begin{lemma}
\label{lem:gamma_stability}
Let the conditions of Theorem~\ref{thm:qp_local_lipschitz} hold, and for each \(i\in\mathcal V\) define \(\chi^i \eqdef \operatorname{col}(\xi_1^i,\xi_2^i,\xi_3^i,\xi_4^i,\zeta_1^i,\zeta_2^i,\zeta_3^i,\zeta_4^i)\in\mathbb R^8\), with \(\chi \eqdef \operatorname{col}(\chi^1,\dots,\chi^N)\in\mathbb R^{8N}\). Then, along closed-loop trajectories of the decentralized QP controller~\eqref{eq:tfl_safe_pf_qp} contained in \(\bar\Omega\), the transverse coordinates satisfy
$\dot\chi = A_\perp\chi,$
where \(A_\perp\in\mathbb R^{8N\times8N}\) is Hurwitz. Hence \(\chi(t)\) converges exponentially to zero. Since \(\chi=0\) locally characterizes the path-following manifold, \(\mathcal P_\Omega\) is locally exponentially attractive for closed-loop trajectories contained in \(\bar\Omega\). Along the same trajectories, the heading error \(\beta_\psi^i=\mu_1^i-\psi^i_{\text{des}}\) of every agent \(i\in\mathcal V\) converges exponentially to zero.
\end{lemma}

\begin{proof}
Along closed-loop trajectories contained in \(\bar\Omega\), the controller is well defined and locally Lipschitz by Theorem~\ref{thm:qp_local_lipschitz}, and the hard transverse equalities in~\eqref{eq:relaxed_tfl_constraint} preserve the nominal transverse dynamics: the relaxation \(E\delta^i\) affects only the speed row, so the \(\xi\)- and \(\zeta\)-subsystems remain exactly the nominal transverse subsystems~\eqref{eq:linear_systems} with the feedback~\eqref{eq:transversal_in}. Therefore, for each \(i\in\mathcal V\), \(\dot \chi^i=A_\perp^i\chi^i\), where \(A_\perp^i\in\mathbb R^{8\times8}\) is Hurwitz by construction, and stacking all agents gives \(\dot\chi=A_\perp\chi\), where \(A_\perp\eqdef\operatorname{blkdiag}(A_\perp^1,\dots,A_\perp^N)\) is Hurwitz.  Since the coordinate transformation~\eqref{eq:transform} is a local diffeomorphism at every point of \(\Omega\), \(\chi\) is a smooth function on \(\Omega\), and for \(\mathbf x\in\Omega\) we have \(\chi=0\) if and only if \(\mathbf x\in\mathcal P_\Omega\). So \(\chi=0\) locally characterizes the path-following manifold on the operating neighborhood. Likewise, the relaxation leaves the \(\mu\)-subsystem under the nominal tangential feedback \(v_\mu^i\) in~\eqref{eq:tang_in}, whose characteristic polynomial is Hurwitz, so \(\operatorname{col}(\beta_\psi^i,\mu_2^i)\) obeys linear, exponentially stable dynamics.
\end{proof}

 Let \(P_\perp\succ0\) solve the Lyapunov equation \(A_\perp^\top P_\perp+P_\perp A_\perp=-I\), with the Hurwitz matrix \(A_\perp\) of Lemma~\ref{lem:gamma_stability}, and define \(V(\mathbf x)\eqdef \chi^\top P_\perp \chi\). Then $V$ is locally positive definite with respect to
$\mathcal P_\Omega$ and, by Lemma~\ref{lem:gamma_stability}, satisfies $\dot V=-\|\chi\|^2\le0$ along closed-loop solutions
in $\bar\Omega$. Define the sublevel set
$\Omega_c\eqdef\{\mathbf x\in\Omega:V(\mathbf x)<c\}$, with $c$ small
enough that the output of every agent in $\Omega_c$ stays strictly
inside its tube $U_\gamma^i$. Such a $c$ exists because we require
$\|s^i\|$ to be bounded away from zero near the boundary of
$U_\gamma^i$, a mild condition on the path representation. Define the admissible set
$\Omega_0\eqdef\Omega_c\cap\bigcap_{i,\rho,\pm}\widetilde{\mathcal B}_{\rho}^{i\pm}\cap\bigcap_{\{i,j\}}\bar{\mathcal B}_{ij}$. Let
$\mathbf x(0)\in\Omega_0$ and let $t^\star\le\infty$ be the first time
the closed-loop solution leaves $\Omega_c\cap\bar\Omega$. On
$[0,t^\star)$ the controller is well defined and locally Lipschitz by
Theorem~\ref{thm:qp_local_lipschitz}, $\dot V=-\|\chi\|^2\le0$, and
the QP of each agent enforces every ECBF inequality
in~\eqref{eq:ineq_const}. Suppose $t^\star$ is finite. The solution
can then exit only through the boundary $V=c$, through the boundary of
an ECBF admissible set, or through the boundary of $\Omega$ (cf.\ the
sublevel-set argument in~\cite[Thm.~V.2]{el2007passivity}). The first
exit is impossible because $V$ is nonincreasing on $[0,t^\star)$ and
$V(\mathbf x(0))<c$. The second is impossible because the comparison argument underlying
Theorem~\ref{thm:ecbf} applies on any interval where the constraints
hold along the solution.

By construction
of the sets $U^i$, the boundary of $\Omega$ is reached only
through the position, the roll and pitch, or the thrust. Here we take
each $U^i$ large enough to contain every state whose output lies
strictly inside $U_\gamma^i$, whose roll and pitch satisfy
$|\phi^i|,|\theta^i|\le\tfrac{\pi}{2}-\epsilon$, and whose thrust is
positive. Such a choice is possible whenever $\det D^i\neq0$ and the
relative-degree condition hold on that set.  For the paths of Section~\ref{sec:sim}, \(\det D^i=2(x_{13}^i)^2\cos\phi^i/(I_xI_yI_zm^3\cos\theta^i)\) at every position other than the circle center, which lies outside the tube. Hence \(D^i\) is singular only at zero thrust and at \(\phi^i=\pm\pi/2\), and undefined only at \(\theta^i=\pm\pi/2\). For other path representations, this condition must be verified directly. The position
stays strictly inside the tube because $V<c$, the roll and pitch stay
inside $\pm(\tfrac{\pi}{2}-\epsilon)$ by the attitude barriers, and
Assumption~\ref{ass:non_zero_thrust} keeps $x_{13}^i>0$, so the third
exit is impossible as well. Hence no finite first exit time exists.
Assuming the closed-loop solution is forward complete, a hypothesis we
also place in Theorem~\ref{thm:main_safety_only} because the analysis does not bound the along-path states,
$\mathbf x(0)\in\Omega_0$ implies
$\mathbf x(t)\in\Omega_c\cap \bar\Omega$ for all $t\ge0$. Without
forward completeness, every statement holds on the maximal interval of
existence. For non-compact paths, these bounds must also hold
uniformly along the path, as in~\cite{eldering2013noncompact}. We now state the main safety
guarantee.

\begin{theorem}
\label{thm:main_safety_only}
Consider the multi-quadrotor system with each agent \(i\in\mathcal V\) applying the decentralized QP controller~\eqref{eq:tfl_safe_pf_qp}, equivalently~\eqref{eq:reconstruction_u_star}, under Assumptions~\ref{ass:non_zero_thrust} and~\ref{ass:reg_slater}.  If \(\mathbf x(0)\in\Omega_0\) and the closed-loop solution is forward complete, then all agents avoid the attitude singularities with margin \(\epsilon\) for all \(t\ge0\), and \(\|x_q^i(t)-x_q^j(t)\|\ge d_s\) for all \(t\ge0\) and every unordered pair \(\{i,j\}\subset\mathcal V\).
\end{theorem}

\begin{proof}
Since \(\mathbf x(0)\in\Omega_0\) and the closed-loop solution is forward complete, the argument preceding the theorem gives \(\mathbf x(t)\in\bar\Omega\) for all \(t\ge0\), where each feedback map \(\mathbf x\mapsto(\mathbf u^i)^\star(\mathbf x)\) is locally Lipschitz by Theorem~\ref{thm:qp_local_lipschitz}.
The QP of agent \(i\) enforces~\eqref{eq:attitude_ecbf} for every \(\rho\in\{\phi,\theta\}\) and each sign, and \(\mathbf x(0)\in\Omega_0\subseteq\widetilde{\mathcal B}_{\rho}^{i\pm}\) gives \(\mathbf x^i(0)\in\mathcal B_{\rho}^{i\pm}\). The comparison argument of Theorem~\ref{thm:ecbf}, applied along the solution in \(\bar\Omega\) as in the argument preceding the theorem, therefore renders each corresponding admissible set forward invariant, so  \(h_{\rho}^{i\pm}(\mathbf x^i(t))\ge0\) for all \(t\ge0\) and \(\rho\in\{\phi,\theta\}\). By~\eqref{eq:att_cbfs}, \(|\rho^i(t)|\le a_\epsilon=\tfrac{\pi}{2}-\epsilon\) for \(\rho\in\{\phi,\theta\}\), so every agent avoids the attitude singularities with margin \(\epsilon\).
Finally, fix any unordered pair \(\{i,j\}\subset \mathcal V\), and define the aggregate feedback law
\(\kappa(\mathbf x):=\operatorname{col}\big((\mathbf u^1)^\star(\mathbf x),\dots,(\mathbf u^N)^\star(\mathbf x)\big)\).
By Theorem~\ref{thm:qp_local_lipschitz}, \(\kappa\) is locally Lipschitz on \(\bar\Omega\). The local QP of agent \(i\) enforces~\eqref{eq:collision_ecbf_constraint}, and the local QP of agent \(j\) enforces the analogous inequality obtained by interchanging \(i\) and \(j\). Adding these two inequalities yields the ECBF condition~\eqref{eq:centralized_col_constraint} for \(\bar h_{ij}\) along the full closed-loop multi-agent system~\eqref{eq:full_dynamics}. Since \(\mathbf x(0)\in \Omega_0\subseteq \bar{\mathcal B}_{ij}\), the same comparison argument, applied to the aggregate feedback \(\kappa\), implies that \(\bar h_{ij}(\mathbf x(t))\ge 0\), that is, \(\|x_q^i(t)-x_q^j(t)\|\ge d_s\) by~\eqref{eq:Bij_clean}, for all \(t\ge 0\).
Since \(\{i,j\}\) was arbitrary, pairwise collision avoidance holds for all unordered pairs.
\end{proof}

 Since every closed-loop trajectory starting in \(\Omega_0\) remains in \(\bar\Omega\), Lemma~\ref{lem:gamma_stability} applies along all such trajectories, and together the two results solve Problem~\ref{prob:safe_pf} under Assumptions~\ref{ass:non_zero_thrust} and~\ref{ass:reg_slater} and the forward-completeness hypothesis of Theorem~\ref{thm:main_safety_only}, with the admissible initial conditions given by \(\Omega_0\). As in other CBF-based schemes, safety does not imply liveness, that is, continued progress along the path, and symmetric encounters can in principle bring agents to a permanent stop~\cite{wang2017safety,grover2023deadlock}.

\begin{remark}\label{rem:cascade}
  If all four TFL equalities in~\eqref{eq:relaxed_tfl_constraint} are relaxed independently,
so that \(E=I_4\) and \(\bm\delta^i\in\mathbb R^4\), and the objective is
chosen to minimize
\(\tfrac12\|\mathbf u^i-\mathbf u_{\mathrm{TFL}}^i\|_{W_i}^2\),
then the equality constraint only defines
\(\bm\delta^i=D^i\mathbf u^i-\mathbf v^i\). The resulting QP reduces
exactly to the nominal TFL controller followed by the ECBF safety
filter~\eqref{eq:qp_safety_filter}. If instead \(E=\col(0,0,1,0)\) and \(W_i=0\), then \(q^i=0\) in~\eqref{eq:reduced_qp}, so \((\delta^i)^\star\) is the feasible slack of least magnitude. It is zero, and~\eqref{eq:tfl_controller} is applied unchanged, whenever the TFL input satisfies~\eqref{eq:ineq_const}. Moreover, since both transversal rows remain hard, collision avoidance cannot induce an off-path vertical maneuver that exposes a lower quadrotor to the downwash of an upper one~\cite{wang2017safe}.
\end{remark}

\begin{figure}[htbp]
    \centering
    \includegraphics[width=0.9\columnwidth]{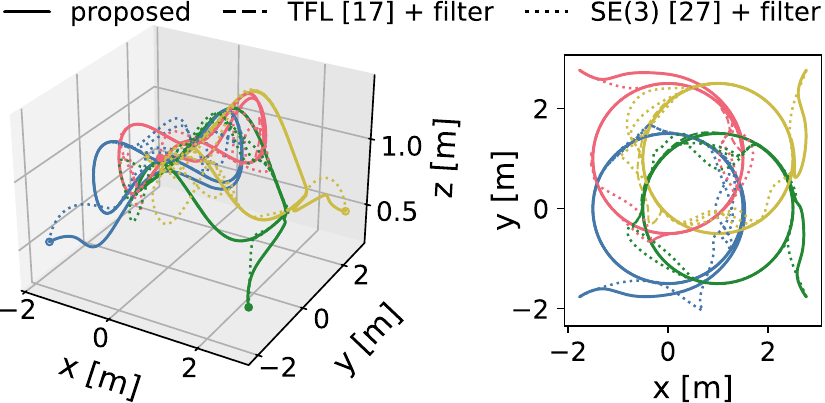}
    \vspace{-0.1in}
    \caption{Path following in Drake for four quadrotors on intersecting circular paths: proposed (solid), TFL~\cite{akhtar2012path}\,$+$\,filter (dashed), SE(3)~\cite{lee2010geometric}\,$+$\,filter (dotted). \textbf{Left:} three-dimensional trajectories. \textbf{Right:} top view. From off-path initial conditions, the proposed controller converges to and, after $15$\,s, remains on its assigned path to within $0.5$\,mm, while both cascades deviate from theirs to avoid collisions.}
    \label{fig:track_circ}
    \vspace{-0.2in}
\end{figure}

\begin{figure*}[t]
    \centering
    \includegraphics[width=\linewidth]{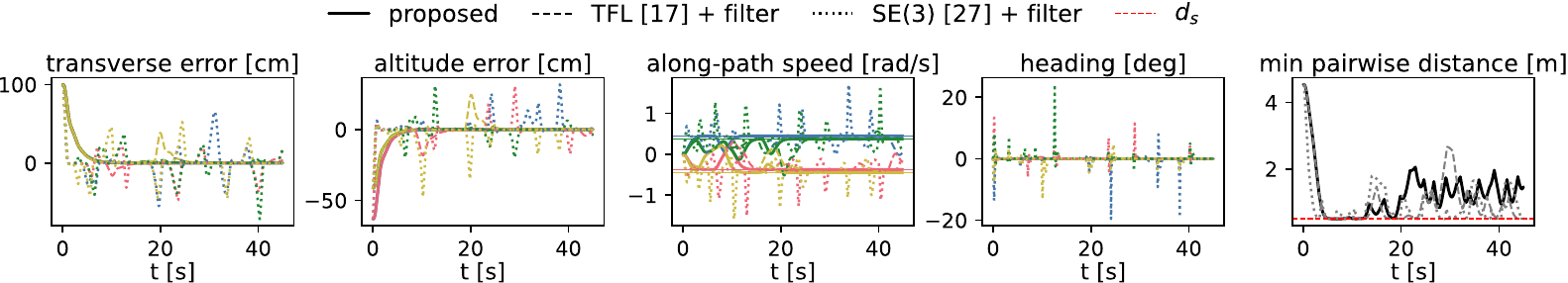}

    \vspace{-0.14in}
    \caption{Output channels and minimum pairwise distance for the circle scenario: proposed (solid), TFL~\cite{akhtar2012path}\,$+$\,filter (dashed), SE(3)~\cite{lee2010geometric}\,$+$\,filter (dotted). The proposed controller drives the transverse and altitude channels to their references and resolves conflicts through the along-path speed, while both cascades perturb every channel.  }
        \label{fig:channels}
    \vspace{-0.2in}
\end{figure*}

\section{Simulation Results}\label{sec:sim}
We evaluate the QP controller~\eqref{eq:tfl_safe_pf_qp} in the rigid-body physics engine \emph{Drake}~\cite{drake}, with mass $m=1.923\,\mathrm{kg}$ and inertia $\operatorname{diag}(I_x,I_y,I_z)=(1.15,\,1.15,\,2.18)\!\times\!10^{-2}\,\mathrm{kg\,m^2}$. Two cascades serve as baselines on identical scenarios. The first, denoted TFL~\cite{akhtar2012path}\,$+$\,filter, passes the nominal TFL controller~\eqref{eq:tfl_controller} of~\cite{akhtar2012path} through the ECBF safety filter~\eqref{eq:qp_safety_filter}, the special case $E=I_4$ of Remark~\ref{rem:cascade}.  The second, SE(3)~\cite{lee2010geometric}\,$+$\,filter, applies the same collision barrier to the commanded acceleration of the geometric controller on SE(3) of~\cite{lee2010geometric}. At this level the barrier has relative degree~2, so the filter uses a double pole at $-3.75$ instead of the quadruple pole of the other two controllers, and it has no attitude barriers. All three controllers share the model, the collision barrier, the safety distance, and the desired speeds. The proposed controller uses $W_i=I$, $P_i=100$, singularity margin $\epsilon=0.1$\,rad, attitude ECBF poles at $-10$, and collision ECBF poles at $-3.75$. These weights give \(P_i/Q^i\ge0.97\), so the speed bias that \(W_i=0\) would remove (Remark~\ref{rem:cascade}) stays below \(7\times10^{-4}\)\,rad/s in median during conflict-free flight. It runs at $400$\,Hz using the clipping rule of Remark~\ref{rem:interval}, without a numerical QP solver. This implementation is a sampled approximation of the continuous-time controller~\eqref{eq:tfl_safe_pf_qp}, and we claim no separate sampled-data invariance guarantee.

The first scenario places four quadrotors on intersecting circles of radius $1.5\,\mathrm{m}$, centered at $(0,0)$, $(0,1)$, $(1,0)$, and $(1,1)\,\mathrm{m}$, with safety distance $d_s=0.5\,\mathrm{m}$. For agent $i$ with circle center $(x_c^i,y_c^i)$, the virtual output is $\hat{\mathbf y}^i=\operatorname{col}\big((x^i-x_c^i)^2+(y^i-y_c^i)^2-1.5^2,\ z^i-z_{\mathrm{des}}(x^i),\ \operatorname{atan2}(y^i-y_c^i,\,x^i-x_c^i), \psi^i\big)$, where $z_{\mathrm{des}}(x)=1.0+0.25\sin(2\pi x/3)$ and $v_{\mathrm{des}}=(0.45,\,-0.375,\,0.375,\,-0.45)\,\mathrm{rad/s}$. The third component is a smooth local path coordinate in the sense of Section~\ref{sec:prelim}: it restricts to the path parameter on the circle, and its relative degree is verified by direct computation.

Fig.~\ref{fig:track_circ} shows the circle-scenario trajectories. From off-path initial conditions, the proposed controller converges to its assigned path and, for $t\ge15$\,s, its maximum distance to the path stays below $0.5$\,mm. Over the same window, the cascade TFL~\cite{akhtar2012path}\,$+$\,filter deviates from its path by up to $40$\,cm, and SE(3)~\cite{lee2010geometric}\,$+$\,filter by up to $71$\,cm. The minimum pairwise distance never falls below $d_s$ for the proposed controller and for TFL~\cite{akhtar2012path}\,$+$\,filter, while SE(3)~\cite{lee2010geometric}\,$+$\,filter dips $7$\,mm inside $d_s$. This occurs because the acceleration-level certificate on~\cite{lee2010geometric} covers the commanded acceleration, which the attitude loop realizes with a lag. Along every run of the proposed controller, the reduced QP~\eqref{eq:reduced_qp} remained strictly feasible at each step, and the collective thrust stayed positive, so Assumption~\ref{ass:non_zero_thrust} held. In the circle scenario, the feasible interval of Remark~\ref{rem:interval} never narrowed below twice the largest slack applied, the thrust never fell below $17.8$\,N, $95\%$ of hover, and \(|\eta_2^i|\) stayed below $0.46$\,rad/s. The initial states were verified to lie in the ECBF admissible sets, but, because the runs start off the path, they need not satisfy the sublevel condition in \(\Omega_0\). The runs therefore demonstrate safe behavior on a larger domain than the certified set, with every barrier condition satisfied throughout.

Fig.~\ref{fig:channels} shows, for the circle scenario, the radial distance of each agent to its circle, its altitude error \(z^i-z_{\mathrm{des}}(x^i)\), its along-path speed \(\eta_2^i\), and its heading \(\psi^i\), with \(\psi^i_{\mathrm{des}}=0\), together with the minimum pairwise distance. The proposed controller holds the transverse, altitude, and heading channels at their references and absorbs every conflict in the along-path speed. Its heading error stays below $10^{-3}$\,degrees, consistent with Lemma~\ref{lem:gamma_stability}, while the cascades perturb the heading by up to $1.4^\circ$ and $19.6^\circ$, with visible yaw oscillation at the encounters. A scenario with two quadrotors on intersecting non-planar sinusoidal paths and animations are at \url{https://gradslab.github.io/safe_multiquad_pf/}.

\section{Conclusion}\label{sec:conclusion}

This paper presented a decentralized safe path-following controller for multiple quadrotors on intersecting paths. By reformulating the transverse feedback linearization (TFL) controller as a QP, the proposed method preserves the transversal path-following dynamics and the nominal heading-error dynamics while incorporating ECBF-based safety constraints by relaxing only the speed constraint. Simulations in Drake showed that the proposed controller resolves conflicts at intersections through the speed alone, whereas the safety-filtered baselines leave their assigned paths and perturb their heading. Future work will incorporate actuator constraints and design a fallback controller for the instances where Assumption~\ref{ass:reg_slater} fails.

\textit{Disclosure of Generative AI Use: } The Large Language Model (LLM)~\cite{openai_chatgpt_2026} has been used to improve the syntax and grammar of some paragraphs in the manuscript.

\bibliographystyle{IEEEtran}
\bibliography{references}

@ARTICLE{beul2018fast,
  author={Beul, Marius and Droeschel, David and Nieuwenhuisen, Matthias and Quenzel, Jan and Houben, Sebastian and Behnke, Sven},
  journal={IEEE Robotics and Automation Letters},
  title={{Fast Autonomous Flight in Warehouses for Inventory Applications}},
  year={2018},
  volume={3},
  number={4},
  pages={3121--3128},
  doi={10.1109/LRA.2018.2849833}}

@ARTICLE{silano2021power,
  author={Silano, Giuseppe and Baca, Tomas and Penicka, Robert and Liuzza, Davide and Saska, Martin},
  journal={IEEE Robotics and Automation Letters},
  title={{Power Line Inspection Tasks With Multi-Aerial Robot Systems Via Signal Temporal Logic Specifications}},
  year={2021},
  volume={6},
  number={2},
  pages={4169--4176},
  doi={10.1109/LRA.2021.3068114}}

@article{nielsen2008local,
author = {Nielsen, Christopher and Maggiore, Manfredi},
title = {{On Local Transverse Feedback Linearization}},
journal = {SIAM Journal on Control and Optimization},
volume = {47},
number = {5},
pages = {2227--2250},
year = {2008},
doi = {10.1137/070682125}
}

@article{mansouri2018cooperative,
title = {Cooperative coverage path planning for visual inspection},
journal = {Control Engineering Practice},
volume = {74},
pages = {118--131},
year = {2018},
doi = {10.1016/j.conengprac.2018.03.002},
author = {Sina Sharif Mansouri and Christoforos Kanellakis and Emil Fresk and Dariusz Kominiak and George Nikolakopoulos},
}

@INPROCEEDINGS{Cichella2012cdc,
  author={Cichella, Venanzio and Kaminer, Isaac and Xargay, Enric and Dobrokhodov, Vladimir and Hovakimyan, Naira and Aguiar, A. Pedro and Pascoal, António M.},
  booktitle={2012 51st IEEE Conference on Decision and Control (CDC)},
  title={{A Lyapunov-based approach for Time-Coordinated 3D Path-Following of multiple quadrotors}},
  year={2012},
  pages={1776--1781},
  doi={10.1109/CDC.2012.6425933}}

@INPROCEEDINGS{Mehdi2017cdc,
  author={Mehdi, Syed Bilal and Cichella, Venanzio and Marinho, Thiago and Hovakimyan, Naira},
  booktitle={2017 IEEE 56th Annual Conference on Decision and Control (CDC)},
  title={Collision avoidance in multi-vehicle cooperative missions using speed adjustment},
  year={2017},
  pages={2152--2157},
  doi={10.1109/CDC.2017.8263963}}

@ARTICLE{xie2020cooperative,
  author={Xie, Wei and Cabecinhas, David and Cunha, Rita and Silvestre, Carlos},
  journal={IEEE Transactions on Systems, Man, and Cybernetics: Systems},
  title={{Cooperative Path Following Control of Multiple Quadcopters With Unknown External Disturbances}},
  year={2022},
  volume={52},
  number={1},
  pages={667--679},
  doi={10.1109/TSMC.2020.3032401}}

@INPROCEEDINGS{ames2019control,
  author={Ames, Aaron D. and Coogan, Samuel and Egerstedt, Magnus and Notomista, Gennaro and Sreenath, Koushil and Tabuada, Paulo},
  booktitle={2019 18th European Control Conference (ECC)},
  title={{Control Barrier Functions: Theory and Applications}},
  year={2019},
  pages={3420--3431},
  doi={10.23919/ECC.2019.8796030}}

@INPROCEEDINGS{nguyen2016exponential,
  author={Nguyen, Quan and Sreenath, Koushil},
  booktitle={2016 American Control Conference (ACC)},
  title={{Exponential Control Barrier Functions for enforcing high relative-degree safety-critical constraints}},
  year={2016},
  pages={322--328},
  doi={10.1109/ACC.2016.7524935}}

@ARTICLE{wang2017safety,
  author={Wang, Li and Ames, Aaron D. and Egerstedt, Magnus},
  journal={IEEE Transactions on Robotics},
  title={{Safety Barrier Certificates for Collisions-Free Multirobot Systems}},
  year={2017},
  volume={33},
  number={3},
  pages={661--674},
  doi={10.1109/TRO.2017.2659727}}

@ARTICLE{yao2022guiding,
  author={Yao, Weijia and de Marina, H{\'e}ctor Garc{\'i}a and Sun, Zhiyong and Cao, Ming},
  journal={IEEE Transactions on Robotics},
  title={Guiding Vector Fields for the Distributed Motion Coordination of Mobile Robots},
  year={2023},
  volume={39},
  number={2},
  pages={1119--1135},
  doi={10.1109/TRO.2022.3224257}}

@INPROCEEDINGS{wang2017safe,
  author={Wang, Li and Ames, Aaron D. and Egerstedt, Magnus},
  booktitle={2017 IEEE International Conference on Robotics and Automation (ICRA)},
  title={Safe certificate-based maneuvers for teams of quadrotors using differential flatness},
  year={2017},
  pages={3293--3298},
  doi={10.1109/ICRA.2017.7989375}}

@INPROCEEDINGS{goarin2025decentralized,
  author={Goarin, Manohari and Li, Guanrui and Saviolo, Alessandro and Loianno, Giuseppe},
  booktitle={2025 IEEE International Conference on Robotics and Automation (ICRA)},
  title={{Decentralized Nonlinear Model Predictive Control for Safe Collision Avoidance in Quadrotor Teams with Limited Detection Range}},
  year={2025},
  pages={5387--5393},
  doi={10.1109/ICRA55743.2025.11128389}}

@INPROCEEDINGS{li2020fast,
  author={Li, Zhichao and Arslan, Omur and Atanasov, Nikolay},
  booktitle={2020 IEEE International Conference on Robotics and Automation (ICRA)},
  title={{Fast and Safe Path-Following Control using a State-Dependent Directional Metric}},
  year={2020},
  pages={6176--6182},
  doi={10.1109/ICRA40945.2020.9197377}}

@article{NieMag04,
title = {{Maneuver regulation via transverse feedback linearization: Theory and examples}},
journal = {IFAC Proceedings Volumes},
volume = {37},
number = {13},
pages = {57--64},
year = {2004},
note = {6th IFAC Symposium on Nonlinear Control Systems 2004 (NOLCOS 2004), Stuttgart, Germany, 1-3 September, 2004},
doi = {10.1016/S1474-6670(17)31200-4},
author = {Chris Nielsen and Manfredi Maggiore},
}

@ARTICLE{AkhWasNie2013Journal,
  author={Akhtar, Adeel and Nielsen, Christopher and Waslander, Steven L.},
  journal={IEEE Transactions on Robotics},
  title={{Path Following Using Dynamic Transverse Feedback Linearization for Car-Like Robots}},
  year={2015},
  volume={31},
  number={2},
  pages={269--279},
  doi={10.1109/TRO.2015.2395711}}

@INPROCEEDINGS{akhtar2012path,
  author={Akhtar, Adeel and Waslander, Steven L. and Nielsen, Christopher},
  booktitle={2012 51st IEEE Conference on Decision and Control (CDC)},
  title={Path following for a quadrotor using dynamic extension and transverse feedback linearization},
  year={2012},
  pages={3551--3556},
  doi={10.1109/CDC.2012.6425945}}

@INPROCEEDINGS{roza2012path,
  author={Roza, Ashton and Maggiore, Manfredi},
  booktitle={2012 American Control Conference (ACC)},
  title={Path following controller for a quadrotor helicopter},
  year={2012},
  pages={4655--4660},
  doi={10.1109/ACC.2012.6315061}}

@misc{drake,
  title={{Drake: Model-based design and verification for robotics}},
  author={Tedrake, Russ and others},
  year={2019},
  note={\url{https://drake.mit.edu}}
}

@INPROCEEDINGS{bouabdallah2005backstepping,
  author={Bouabdallah, S. and Siegwart, R.},
  booktitle={Proceedings of the 2005 IEEE International Conference on Robotics and Automation},
  title={{Backstepping and Sliding-mode Techniques Applied to an Indoor Micro Quadrotor}},
  year={2005},
  pages={2247--2252},
  doi={10.1109/ROBOT.2005.1570447}}

@book{eldering2013noncompact,
  author    = {Eldering, Jaap},
  title     = {{Normally Hyperbolic Invariant Manifolds: The Noncompact Case}},
  series    = {Atlantis Studies in Dynamical Systems},
  volume    = {2},
  publisher = {Atlantis Press},
  year      = {2013}
}

@article{mestres2025regularity,
title = {Regularity properties of optimization-based controllers},
journal = {{European Journal of Control}},
volume = {81},
pages = {101098},
year = {2025},
doi = {10.1016/j.ejcon.2024.101098},
author = {Pol Mestres and Ahmed Allibhoy and Jorge Cortés},
}

@book{boyd2004convex,
  author    = {Boyd, Stephen and Vandenberghe, Lieven},
  title     = {Convex Optimization},
  publisher = {Cambridge University Press},
  year      = {2004}
}

@ARTICLE{mestres2024distributed,
  author={Mestres, Pol and Nieto-Granda, Carlos and Cortés, Jorge},
  journal={IEEE Robotics and Automation Letters},
  title={{Distributed Safe Navigation of Multi-Agent Systems Using Control Barrier Function-Based Controllers}},
  year={2024},
  volume={9},
  number={7},
  pages={6760--6767},
  doi={10.1109/LRA.2024.3414268}}

@INPROCEEDINGS{el2007passivity,
  author={El-Hawwary, Mohamed I. and Maggiore, Manfredi},
  booktitle={2007 46th IEEE Conference on Decision and Control},
  title={Passivity-based stabilization of non-compact sets},
  year={2007},
  pages={1734--1739},
  doi={10.1109/CDC.2007.4434732}}

@article{grover2023deadlock,
author = {Jaskaran Grover and Changliu Liu and Katia Sycara},
title ={{The Before, During, and After of Multi-robot Deadlock}},
journal = {The International Journal of Robotics Research},
volume = {42},
number = {6},
pages = {317--336},
year = {2023},
doi = {10.1177/02783649221074718}
}

@INPROCEEDINGS{lee2010geometric,
  author={Lee, Taeyoung and Leok, Melvin and McClamroch, N. Harris},
  booktitle={2010 49th IEEE Conference on Decision and Control (CDC)},
  title={{Geometric tracking control of a quadrotor UAV on SE(3)}},
  year={2010},
  pages={5420--5425},
  doi={10.1109/CDC.2010.5717652}}

@misc{openai_chatgpt_2026,
  author = {{OpenAI}},
  title  = {{ChatGPT} ({GPT-5.4})},
  year   = {2026},
  note   = {{Large language model, released Mar.~5, 2026}},
  url    = {https://chatgpt.com/}
}

\end{document}